\date{}
\documentclass[11pt]{article}

\usepackage{tikz}

\usepackage{amsmath,amsthm,amsfonts,amssymb,amsopn,amscd} 
\usepackage{color}

\def\b{\beta}

\def\A{{\cal A}}

\def\M{{\cal M}}

\def\R{{\cal R}}

\def\H{{\cal H}}

\def\S{{\cal S}}

\def\f{{\varphi}}

\def\PSL{{{\rm PSL}(2,\mathbb R)}}

\def\S2{S^{1(2)}}

\newtheorem{theorem}{Theorem}[section]
\newtheorem{lemma}[theorem]{Lemma}

\newtheorem{proposition}[theorem]{Proposition}

\theoremstyle{definition} 

\theoremstyle{remark} 

\newcommand{\ben}{\begin{equation}}
\newcommand{\een}{\end{equation}}

\def\PSL{PSU(1,1)}

\def\SL2{{{\rm SL}(2,\R)}}

\def\PSL2{{{\rm PSL}(2,\Reali)}}

\def\U1{{{\rm V}(1)}}
\def\SU2{{{\rm SV}(2)}}

\def\SU{{{\rm SU}}}

\def\A{{\mathcal A}}

\def\H{{\mathcal H}}

\def\M{{\mathcal M}}

\title{\Huge{Entropy of coherent excitations }}

\author{{\sc Roberto Longo\thanks{Supported by the ERC Advanced Grant 669240 QUEST ``Quantum Algebraic Structures and Models'', MIUR FARE R16X5RB55W  QUEST-NET and GNAMPA-INdAM.}
}
\\
Dipartimento di Matematica,
Universit\`a di Roma ``Tor Vergata'',\\
Via della Ricerca Scientifica, 1, I-00133 Roma, Italy\\
E-mail: {\tt longo@mat.uniroma2.it}
}

\date{}
\begin{document}

\maketitle

\begin{abstract}
We provide a rigorous, explicit formula for the vacuum relative entropy of a coherent state on wedge local von Neumann algebras associated with a free, neutral quantum field theory on the Minkowski spacetime of arbitrary spacetime dimension. We consider charges localised on the time zero hyperplane, possibly crossing the boundary. 
\end{abstract}

\newpage

\section{Introduction}
Recently, much attention has been focussed on quantum information aspects of Quantum Field Theory, an interest that can however be traced back to Black Hole Thermodynamics (see e.g. \cite{Wa}). While the physical literature on the subject is extensive, the mathematical foundations are less studied although, in important cases, the mathematical analysis has led to the right physical framework. This is the case, for example, of the Tomita-Takesaki modular theory of von Neumann algebras (see \cite{T}), which is linked to the KMS thermal equilibrium condition (see \cite{H}) and is essential in this paper.  

We do not attempt here to give an even schematic sketch of the past and actual research on the subject, we rather refer to \cite{L18}, and reference therein, for a rigorous approach to the subject. Indeed, this paper is a continuation of the work done in \cite{L18}, especially in Sect. 4, where a first and detailed analysis of the vacuum relative entropy of a localised state has been provided for the chiral, conformal net of von Neumann algebras associated with $U(1)$-current. Here, we analyse the case of a free scalar field of any spacetime dimension $d+1$ 
(but for the massless $d=1$ case, see however footnote 1). While in the $U(1)$-current case localised states may lead to different representations, in the higher dimensional case a localised state is represented by a vector of the vacuum Hilbert space; we shall consider here only coherent (or exponential)  vectors, that however form a total set vectors in the Hilbert space.

Our main result is contained in Theorem \ref{second}. Let $W$ the wedge region $\{x: x_1 > |x_0|\}$ of the Minkowski spacetime. Let $\f$ be the vacuum state, and $\f_{h_\phi+ k_\pi }$
the coherent state associated with a time-zero field and momentum vector given by smooth, real test functions 
$h$ and $k$ on the time-zero hyperplane. The relative entropy $S(\f_{h_\phi + k_\pi} |\!| \f)$ between the restrictions of $\f_{h_\phi + k_\pi}$ and $\f$ to the local von Neumann algebra $\A(W)$ associated with $W$ is given by:
\[
S(\f_{h_\phi + k_\pi} |\!| \f)  =  \pi\int_{x_1 > 0} x_1   h^2({\bf x})d{\bf x} +  
\pi\int_{x_1 > 0} x_1 \Big(|\nabla k ({\bf x})|^2 + m^2 k^2({\bf x}) \Big) d{\bf x}  \ .
\]
Here, $d{\bf x} = dx_1dx_2\dots dx_d$ is the space volume. The above formula can be expressed also as a space integration detecting boundary entropy contributions (Lemma \ref{ek2}) that are not visible in the one-dimensional chiral case \cite{L18}. 
 
Possible ways to continue our study are discussed in the outlook. 
Finally, we refer to \cite{CF, HO, OT, X18} for recent literature related to our work.

\section{Second quantisation preliminaries}\label{SQ}

Let $\H$ be a complex Hilbert space and  $\Gamma(\H)$  the {\em exponential} of $\H$, i.e. the Bosonic Fock space over $\H$ (also denoted by $e^\H$). Thus
\[
\Gamma(\H)\equiv \bigoplus^\infty_{n=0}\H_s^{\otimes^n}\ ,
\]
$\H_0\equiv\mathbb C\xi$ is the one-dimensional Hilbert space a unit vector $\xi$ called the {\em vacuum} vector, and $\H_s^{\otimes^n}$ is the symmetric Hilbert $n$-fold tensor product of $\H$. 

If ${h}\in\H$, we denote by $e^{h}$ the coherent vector of $e^\H$:
\[
e^{h}\equiv \bigoplus^\infty_{n=0} \frac{1}{\sqrt{n!}}(h^{\otimes^n})_s
\]
where the zeroth component of $e^h$ is $\xi$, thus $e^0 = \xi$. 
One may check that
\[
(e^{h} ,e^{k}) = e^{({h},{k})}\ ,
\]
and $\{e^{h},\ {h}\in\H\}$ is a total family  of independent vectors of $\Gamma(\H)$.

If $U$ is a unitary on $\H$,  $\Gamma(U)$ (or $e^U$) is the unitary on $\Gamma(\H))$
\[
\Gamma(U) = 1\oplus U \oplus (U\otimes U)\oplus (U\otimes U \otimes U)\oplus\cdots
\]
is the unitary second quantization of $U$; if $U$ is anti-unitary, then $\Gamma(U)$ is the similarly defined anti-unitary on $\Gamma(\H)$. Note that $e^U e^{h} = e^{U{h}}$. 

Setting
\[
V({h})e^{k}\equiv e^{-\frac12 ({h},{h})}e^{-\Re({h},{k})}e^{{h}+{k}}
\]
we get an isometry on  $\{e^{h},\ {h}\in\H\}$, that extends to a unitary operator $V({h})$ on $e^\H$. 
The $V({h})$'s are called {\em Weyl unitaries};  the map ${h}\mapsto V({h})$ is norm - strong operator continuous and gives a representation of the Weyl commutation relations
\begin{equation}\label{Weyl}
V({h} +{k}) = e^{i\Im ({h},{k})}V(h)V({k})\ .
\end{equation}
Note that
\begin{equation}\label{cohvacuum}
V({h})\xi = V({h})e^{0} = e^{-\frac12 ({h},{h})}e^{{h}} \ ,
\end{equation}
therefore
\[
(V(k)\xi, V(h)\xi) 
= e^{-\frac12 (||h||^2 + ||k||^2)}(e^k , e^h) 
= e^{-\frac12 (||h||^2 + ||k||^2)}e^{(k,h)} 
\]
and, in particular,
\ben\label{fV}
\f(V({h})) = e^{-\frac12 ||{h}||^2} \ .
\een
where $\f \equiv (\xi,\cdot \xi)$ is the vacuum state.

By the uniqueness of the GNS representation, the above Fock representation is (up to unitary equivalence) the unique representation of the Weyl commutation relations on a Hilbert space with a cyclic vector $\xi$ such that $(\xi,V({h})\xi) = e^{-\frac12 ||{h}||^2}$.

Let $H\subset\H$ be a real linear subspace. We put
\[
R(H) \equiv \{V({h}):\ {h}\in H\}'' \ ,
\]
namely $R(H)$ is the von Neumann algebra on $e^\H$ given by the weak closure of the linear span of the $V({h})$'s as ${h}$ varies in $H$.

The real linear subspace $H$ is said to be a {\it standard subspace} if $H$ is closed and
\[
\overline{H + iH} = \H\ , \qquad H\cap i H = \{0\}  .
\]
We refer to \cite{T} and \cite{L,LN} for the modular theory of von Neumann algebras and standard subspaces. 
\begin{proposition}\label{R(H)basic}
Let $H$ be a standard subspace. We have:
\begin{itemize}
\item[$(a)$] If $K$ a dense  subspace of $H$, then $R(K)= R(H)$;
\item[$(b)$] $\xi$ is a cyclic and separating vector for $R(H)$;
\item[$(c)$] Then the modular unitaries and conjugation associated with $(R(H),\xi)$ are given by
\[
\Delta^{it}_{R(H)} = \Gamma(\Delta^{it}_{H}), \quad J_{R(H)} = \Gamma(J_{H})\ .
\]
\end{itemize}
\end{proposition}
\noindent
Here, $\Delta_H$ and $J_H$ are the modular operator and the modular conjugation on $\H$ associated with $H$. 

\section{A first general formula}
Recall that the relative entropy between two normal, faithful states $\f_1, \f_2$ of a von Neumann algebra $\M$
is given by Araki's formula  \cite{Ar}
\[
S(\f_1 |\!| \f_2) = - (\xi_1 , \log\Delta_{\xi_2,\xi_1}\xi_1)\ ;
\]
here $\xi_1 , \xi_2$ are any cyclic vector representatives of $\f_1 , \f_2$ on the underlying Hilbert space (that always exist in the standard representation) and $\Delta_{\xi_2,\xi_1}$ is the associated relative modular operator. 

Given a standard subspace $H$ and a vector $h\in H$, we define here the {\it entropy of $h$} with respect to $H$ by
\[
S_h = - (h,\log\Delta_H h) \ .
\]
$S_h = S^H_h$ is strictly positive unless $h=0$ and finite for a dense subset of $H$. If $U$ is a unitary on $\H$, then
\[
S^{UH}_{Uh} = U S^H_h U^*\ .
\]
We shall see that the relative entropy between coherent states is given by entropy of vectors. 

Let $h\in\H$ and $V(h)\xi = e^h/{||e^h||}$ the normalised coherent vector. Suppose $H\subset \H$ is a standard subspace and let $\f_h = (V(h)\xi, \cdot\, V(h)\xi)$ as a state on $R(H)$. Note that
\[
\f_h = \f\cdot{\rm Ad}V(h)^* \big |_{R(H)}\ .
\]
We want to study the relative entropy 
\[
S(\f_h |\!| \f_k)\ , \quad h,k\in H\ ,
\]
between the states $\f_h$ and $\f_k$ of $R(H)$. 
Since
\ben\label{hk}
S(\f_h |\!| \f_k) = S(\f\cdot{\rm Ad}V(-h)V(k) |\!| \f) =  S(\f\cdot{\rm Ad}V(k-h) |\!| \f)
=  S(\f_{k-h} |\!| \f) \ ,
\een
we may restrict our analysis to the case $k=0$, namely $\f_k$ is the vacuum state. 

If $h\in H$, we have
\[
S(\f_h |\!| \f) = - ( \xi , \log\Delta_{V(h)\xi, \xi}\xi)=  - ( V(h)\xi , \log\Delta_{R(H)} V(h)\xi)
\]
with $\Delta_{R(H)}$ the modular operator associated with $(R(H),\xi)$. Therefore
\begin{proposition}
Let $h\in H$. The relative entropy on $R(H)$ between $\f_h$ and $\f$ is given by
\[
S(\f_h |\!| \f) =  - {(h, \log\Delta_H h)} \ .
\]
\end{proposition}
\begin{proof}
\begin{multline*}
S(\f_h |\!| \f) = i\frac{d}{ds}( V(h)\xi , \Delta^{is}_{R(H)} V(h)\xi)\big |_{s=0}
= i\frac{d}{ds}( V(h)\xi , \Gamma(\Delta^{is}_{H}) V(h)\xi)\big |_{s=0} 
\\= i\frac{d}{ds}( V(h)\xi , V(\Delta^{is}_{H}h)\xi)\big |_{s=0} 
 = i e^{- ||h||^2}\frac{d}{ds}e^{(h, {\Delta_H^{is}} h)}\big |_{s=0} 
= - {(h, \log\Delta_H h)}\ .
\end{multline*}
\end{proof}
Note that $S(\f_h |\!| \f)$ is real, thus
\ben\label{Scom}
S(\f_h |\!| \f) = - (h, \log\Delta_H h) = -\Im i(h, \log\Delta_H h) = -\Im (h, i\log\Delta_H h) \ ,
\een
so $S(\f_h |\!| \f)$ may be computed via the symplectic form $\Im(\cdot ,\cdot)$. 

\section{Entropy of local vacuum excitations }
We are going  to consider the Minkowski spacetime $\mathbb R^{d+1}$. 
We shall use the symbol $x$ for a spacetime vector in $\mathbb R^{d+1}$, ${x} = (x_0,x_1,\dots , x_d)$, while a space vector
will be denoted in bold, ${\bf x} = (x_1, \dots , x_d)$. 

Let  $\H= L^2(\mathbb R^{d+1}, d\Omega_m)$ be the one particle Hilbert space of the free scalar Bose field on the Minkowski spacetime $\mathbb R^{d+1}$, space dimension $d >1$, mass $m\geq 0$.\footnote{The case $d= 1$, $m=0$ can be studied by an adaptation of the present discussion as only the derivative of the free field exists here. 
Note however that, in this case, the relative entropy of a null localised charge 
is the sum the relative entropies of the charge chiral components, whose values are given by the $U(1)$-current formula in \cite[Sect. 4]{L18}.}
Here $\Omega_m(p) = \delta(p^2 - m^2)$ is the Lorentz invariant measure on the mass $m$ positive hyperboloid $\mathfrak H_m \subset \mathbb R^{d+1}$. 
If $f$ is a function in the Schwartz space $S(\mathbb R^{d+1})$, we consider $f$ as an element of $\H$ via the embedding
\ben\label{emb}
f \mapsto \sqrt{2\pi}\hat f |_{{\mathfrak H}_m} \  ,
\een
where the Fourier transform is taken w.r.t. the Lorentz signature and is normalised so that Plancherel formula holds.

If $f,g$ are real functions in $S(\mathbb R^{d+1})$ we have
\[
\Im(f,g) = -\frac{i}{2}\int_{\mathfrak H_m}\big(\overline{\hat f({p})}\hat g({p}) - 
\hat f({p})\overline{\hat g({p})}\big)d\Omega_m
\]
More generally, if $f\in S'(\mathbb R^{d+1})$ is a real tempered distribution, we consider $f$ as an element of $\H$ via the embedding \eqref{emb}, provided $\hat f$ restricts to a function on $\mathfrak H_m$ in  $L^2(\mathfrak H_m, d\Omega_m)$. 

Let $h$, $k$ be real functions in $S(\mathbb R^d)$. We consider the distributions $h_\f , k_\pi \in S'(\mathbb R^{d+1})$
\[
h_\phi(x_0,\dots , x_d) =\delta(x_0)h(x_1,\dots , x_d),\quad k_\pi(x_0,\dots , x_d) = \delta'(x_0)k(x_1,\dots , x_d),
\]
that belong to $\H$ (time-zero fields and momenta) and span a dense set of vectors of $\H$. We have
\[
\Im(h_{1\phi}, h_{2\phi})
= \Im(k_{1\pi}, k_{2\pi}) = 0
\]
and
\ben\label{hk'}
\Im(h_\phi, k_\pi) = \frac12\int_{\mathbb R^d} h({\bf x})k({\bf x})d{\bf x} \ .
\een
We shall denote by $\A$ the local net of von Neumann algebras associated with the free scalar field with mass $m\geq 0$. Thus, if $O\subset \mathbb R^{d+1}$, we have
\[
\A(O) = \{V(h): h\in H(O)\}''
\]
where $H(O)$ is the closed, real subspace of the one-particle Hilbert space $\H = L^2(\mathbb R^{d+1}, d\Omega_m)$ given by
\[
H(O) = \{h\in S(\mathbb R^{d+1}),\ h \ {\rm real}, \ {\rm supp}(h)\subset O\}^- 
\]
with the embedding $S(\mathbb R^{d+1})\hookrightarrow L^2(\mathbb R^{d+1}, d\Omega_m)$ given by \eqref{emb}. $H(O)$ is standard if both $O$ and its causal complement $O'$ have non-empty interiors. 

Let $W$ be the wedge $x_1 > |x_0|$ and $\Lambda_W$ the associated boost one-parameter group of transformations of $\mathbb R^{d+1}$. We have 
\[
U\big(\Lambda_W(2\pi s)\big) = \Delta_{H(W)}^{-is}
\]
where $U$ is the unitary representation of the Poincar\'e group on $\H$ and $H(W)\subset \H$ is the standard subspace associated with $W$ \cite{BW}. 

With $f\in S'(\mathbb R^{d+1})$, we set
\[
\partial^W_s f = \frac{d}{ds} f\cdot\Lambda_W(s) \ .
\]
\begin{lemma}\label{Spartial}
If $f\in S'(\mathbb R^{d+1})$, {\rm supp}$(f)\subset W$ and $f\in \H$, the relative entropy between the states $\f$ and $\f_f$ of $\A(W)$ is given by 
\[
S(\f_f |\!| \f) =-2\pi i { (f,  \partial^W_0 f)} = 2\pi  \Im(f,  \partial^W_0 f) \ , 
\]
with $\partial^W_0 = \partial^W_s |_{s=0}$. 
\end{lemma}
\begin{proof}
by Corollary \ref{Scom} we have
\[
S(\f_f |\!| \f) =   i\frac{d}{ds} { (f, \Delta_H^{is} f)} \big|_{s=0}
=  - i\frac{d}{ds} {(f,  f\cdot\Lambda_W(2\pi s))} \big|_{s=0}
=-2\pi i { (f,  \partial^W_0 f)} \ ,
\]
so $-2\pi i (f,  \partial^W_0 f)$ is real, namely $(f,  \partial^W_0 f)$ is purely imaginary, hence $\Im(f,  \partial^W_0 f) = -i(f,  \partial^W_0 f)$ and
\[
S(\f_f |\!| \f)  = 2\pi  \Im(f,  \partial^W_0 f) \ .
\]
\end{proof}
Now, $\Lambda_W(s)$ acts only on the $x_0, x_1$ variables, 
\[
\Lambda_W(s)=\begin{pmatrix}\cosh s &\sinh s\\ \sinh s &\cosh s \end{pmatrix}, \qquad \frac{d}{ds}\Lambda_W(s)\big|_{s=0} = 
\begin{pmatrix}0 &1\\ 1  & 0 \end{pmatrix} \ ,
\]
so
\ben\label{partial}
\partial^W_0 = x_1 \partial_{x_0}  + x_0\partial_{x_1}  \ ,
\een
with $\partial_{x_k}$ the partial derivative w.r.t.  $x_k$. 
\begin{lemma}\label{eh1}
 Let $h\in S(\mathbb R^d)$ be real with support in the half-plane $x_1 \geq 0$. We have
\[
S(\f_{h_\phi} |\!| \f) =  \pi\int_{\mathbb R^d} x_1 h^2({\bf x})dx \ .
\]
\end{lemma}
\begin{proof}
 Note first that, in Fourier transform, we have 
 \[
 \widehat{\partial^W_0 f} =\partial^W_0 \hat f 
 \]
 because $\Lambda_W$ is isometric w.r.t. the Lorentz metric. 
Then
\begin{align*}
(h_\phi,{\partial^W_0 h_\phi}) & = 
 \int_{\mathfrak H_m}\overline{\hat h_\phi({\bf p})}\widehat{\partial^W_0 h_\phi}({\bf p})d\Omega_m \\
&  =\int_{\mathfrak H_m}\overline{\hat h_\phi({\bf p})}\partial^W_0\hat{ h}_\phi({\bf p}) d\Omega_m 
 =  \int_{\mathfrak H_m}\overline{\hat h({\bf p})}\partial^W_0 \hat{ h}({\bf p})d\Omega_m\\
& = \int_{\mathfrak H_m}\overline{\hat h({\bf p})}p_1\partial_{p_0} \hat{ h}({\bf p})d\Omega_m
+  \int_{\mathfrak H_m}\overline{\hat h({\bf p})}p_0\partial_{p_1}\hat{ h}({\bf p})d\Omega_m \\
& =  \int_{\mathfrak H_m}\overline{\hat h({\bf p})}p_0\partial_{p_1}\hat{ h}({\bf p})d\Omega_m
= \frac12\int_{\mathbb R^d}\overline{\hat h({\bf p})}\partial_{p_1}\hat{ h}({\bf p})d{\bf p} \\
& = \frac{i}{2}\int_{\mathbb R^d}x_1 h^2({\bf x})d{\bf x} \ ,
\end{align*}
where the first integral in the third line is zero because $\hat h$ does not depend on $p_0$, so
our lemma is proved.  
\end{proof}
\begin{lemma}\label{ek1}
Let $k\in S(\mathbb R^d)$ be real with support in the half-plane $x_1 \geq 0$. We have
\[
S(\f_{k_\pi} |\!| \f) =  
\pi m^2\int_{ \mathbb R^d}x_1 k^2({\bf x})d{\bf x}  +
\pi \int_{ \mathbb R^d}x_1|\nabla k ({\bf x})|^2 d{\bf x} 
\]
(relative entropy on $\A(W)$).
\end{lemma}
\begin{proof}
We have
\begin{align*}
(k_\pi,{\partial^W_0 k_\pi}) & = 
 \int_{\mathfrak H_m}\overline{\hat k_\pi({\bf p})}\widehat{\partial^W_0 k_\pi}({\bf p})d\Omega_m \\
&  =\int_{\mathfrak H_m}\overline{\hat k_\pi({\bf p})}\partial^W_0\hat{ k}_\pi({\bf p}) d\Omega_m 
 =  \int_{\mathfrak H_m}p_0\overline{\hat k({\bf p})}\partial^W_0 (p_0\hat{ k}({\bf p}))d\Omega_m\\
& = \int_{\mathfrak H_m}p_0\overline{\hat k({\bf p})}p_1\partial_{p_0} (p_0\hat{ k}({\bf p}))d\Omega_m
+  \int_{\mathfrak H_m}p_0\overline{\hat k({\bf p})}p_0\partial_{p_1}(p_0\hat{ k}({\bf p}))d\Omega_m\\
& =  \int_{\mathfrak H_m}p_0p_1\overline{\hat k({\bf p})}\hat{ k}({\bf p})d\Omega_m
+  \int_{\mathfrak H_m}p^3_0\overline{\hat k({\bf p})}\partial_{p_1} \hat{ k}({\bf p})d\Omega_m  \  .
\end{align*}
The first term above, on the last line, vanishes because
\begin{multline*}
\int_{\mathfrak H_m}p_0p_1\overline{\hat k({\bf p})}\hat{ k}({\bf p})d\Omega_m
= \frac12\int_{\mathbb R^d}p_1\overline{\hat k({\bf p})}\hat{ k}({\bf p})d{\bf p}\\
=\frac{i}{2}\int_{\mathbb R^d} k({\bf x})\partial_{x_1}k({\bf x}) d{\bf x} 
=\frac{i}{4}\int_{\mathbb R^d} \partial_{x_1}k^2({\bf x}) d{\bf x} = 0 \ . 
\end{multline*}
Therefore
\begin{multline*}
(k_\pi,{\partial^W_0 k_\pi})= \int_{\mathfrak H_m}p^3_0\overline{\hat k({\bf p})}\partial_{p_1} \hat{ k}({\bf p})d\Omega_m  =
 \frac12\int_{ \mathbb R^d}({\bf p}^2+ m^2)\overline{\hat k({\bf p})}\partial_{p_1} \hat{ k}({\bf p})d{\bf p} \\
= - \frac{i}{2} \int_{ \mathbb R^d}x_1 k({\bf x})(\Delta + m^2){ k({\bf x})}d{\bf x} 
=  \frac{i}{2}\int_{ \mathbb R^d}x_1\big(|\nabla k ({\bf x})|^2 +
 m^2 k^2({\bf x})\big)d{\bf x}  - \frac{i}{2} \int_{ \mathbb R^d} k({\bf x})\partial_{x_1}k({\bf x})d{\bf x} \\
 =  \frac{i}{2}\int_{ \mathbb R^d}x_1\big(|\nabla k ({\bf x})|^2 +
 m^2 k^2({\bf x})\big)d{\bf x}  - \frac{i}2 \int_{ \mathbb R^d} \partial_{x_1}k^2({\bf x})d{\bf x}
=  \frac{i}{2}\int_{ \mathbb R^d}x_1\big(|\nabla k ({\bf x})|^2 +
 m^2 k^2({\bf x})\big)d{\bf x} 
\end{multline*}
with $\Delta = -\sum \partial^2_{x_k}$ the space Laplacian and $\nabla = (\partial_{x_1} , \partial_{x_2}, \dots ,\partial_{x_d})$ 
the gradient. So our statement follows by Lemma \ref{Spartial}. 
\end{proof}
\begin{lemma}\label{add}
Let $h,k\in S(\mathbb R^d)$ be real with support in the half-plane $x_1 >0$. We have
\[
(h_\phi + k_\pi, \partial^W_0 (h_\phi + k_\pi)) = (h_\phi , \partial^W_0 h_\phi )+ (k_\pi, \partial^W_0  k_\pi) \ .
\]
As a consequence, if $h,k$ are supported in $x_1 \geq 0$, then the relative entropies on $\A(W)$ add
\[
S(\f_{h_\phi + k_\pi} |\!| \f) = S(\f_{h_\phi} |\!| \f) + S(\f_{k_\pi} |\!| \f) \ .
\]
\end{lemma}
\begin{proof}
As $\partial^W_0$ is a skew-adjoint operator on $\H$, we have
\[
(h_\phi + k_\pi, \partial^W_0 (h_\phi + k_\pi)) = (h_\phi , \partial^W_0 h_\phi )+ (k_\pi, \partial^W_0  k_\pi) 
 + 2i\Im(k_\pi , \partial^W_0 h_\phi )
\]
 \[
\Im(k_\pi, x_1\partial_{x_0} h_\phi) = \Im(k_\pi, (x_1 h)_\pi) = 0 \ ,
\]
\[
\Im(k_\pi, x_0\partial_{x_1} h_\phi) = \Im(k_\pi, x_0\delta(x_0)\partial_{x_1} h) = 0 
\]as $x_0\delta(x_0) = 0$.  So our statement follows by Lemma \ref{Spartial} and formula \eqref{partial}. 
\end{proof}
By putting together the above results we have:
\begin{theorem}\label{first}
Let $h,k\in S(\mathbb R^d)$ be real and supported in $x_1 \geq 0$. The vacuum relative entropy on $\A(W)$ of the coherent state associated with $h_\phi + k_\pi$ is given by 
\[
S(\f_{h_\phi + k_\pi} |\!| \f) = \pi\int_{\mathbb R^d} x_1 \Big( h^2({\bf x}) +  m^2  k^2({\bf x}) + |\nabla k({\bf x})|^2\Big) d{\bf x} \ .
\]
\end{theorem}

\section{Entropy for charges crossing the boundary}
We now extend the analysis in the previous section for states with non zero density on the boundary. 

With $W$ the wedge $x_1 > |x_0|$ as above, let $f$ be a real smooth function on $\mathbb R^{d+1}$ with compact support and $V(f)$ the associated Weyl unitary. 
As in \cite{L18}, $V(f)$ implements an automorphism $\b_f ={\rm Ad}V(f)$ of $\A(O)$ for any spacetime region $O$.  Indeed 
we have:
\[
\b_f\big(V(g)\big) = e^{i\Im(f,g)}V(g)
\]
and $\b_f$ is normal on $\A(O)$ as it is unitarily implemented. 

We set
\[
{\b_{f}}_+ \equiv \b_f \big |_{\A(W)} \  \ {\rm and}\ \    f_{_W} \equiv \chi_{{}_W} f \ ,
\]
 with $\chi_{{}_W}$ the characteristic function of $W$, and
\[
\f_{f_+} \equiv  \f\cdot {\b_{f}}_+^{-1} \ .
\]
We shall treat the case supp$(f) \subset \overline{W\cup W'}$, namely the support of of $f$ is contained in the region $|x_1|\geq |x_0|$.
Now, Lemma \ref{Spartial} does not directly generalise to
\[
S(\f_{f_+} |\!| \f) =-2\pi i { (f_{_W},  \partial^W_0 f_{_W})} = 2\pi  \Im(f_{_W},  \partial^W_0 f_{_W})
\]
because $f_{_W}$ is not an element of $\H$. We shall however see that the following lemma holds.
\begin{lemma}\label{lemmaS+}
Let $f$ be a real function in $S(\mathbb R^4)$ with support in the region $|x_1|\geq |x_0|$. We have: 
\ben\label{S+}
S(\f_{f} |\!| \f) = 2\pi\Im\int_{\mathfrak H_m}\overline{\hat f_{_W}(p)}\widehat{\partial^W_0 f_{_W}}(p)d\Omega_m
= -2\pi i\int_{\mathfrak H_m}\overline{\hat f_{_W}(p)}\widehat{\partial^W_0 f_{_W}}(p)d\Omega_m \ ,
\een
where the states $\f$ and $\f\cdot\b_f^{-1}$ are restricted to the von Neumann algebra $\A(W)$. 
\end{lemma}
\noindent
Taking for granted formula \eqref{S+} for the moment, we now make the entropy computation for certain charges that are possibly non-zero on the boundary. 
\subsection{Time zero fields}
\begin{lemma}\label{eh2}
 Let $h$ be a real smooth function on $\mathbb R^d$ with compact support. We have
\[
S(\f_{h_\phi} |\!| \f)  
= \pi\int_{x_1 >0} x_1 h^2(x)dx \ ,
\]
where the states $\f_{h_\phi} $ and $ \f $ are restricted to $\A(W)$. 
\end{lemma}
\begin{proof}
We follow the lines of the proof of Lemma \ref{eh1}.  So
\begin{align*}
({h_\phi}_+,{\partial^W_0 {h_\phi}_+}) & = 
 \int_{\mathfrak H_m}\overline{\widehat{h_\phi}_+({\bf p})}\widehat{\partial^W_0 {h_\phi}_+}({\bf p})d\Omega_m \\
&  =\int_{\mathfrak H_m}\overline{{\widehat{h_\phi}_+}({\bf p})}\partial^W_0{\widehat{h_\phi}_+}({\bf p}) d\Omega_m 
 =  \int_{\mathfrak H_m}\overline{\hat h_+({\bf p})}\partial^W_0 \hat{ h}_+({\bf p})d\Omega_m\\
& = \int_{\mathfrak H_m}\overline{{\hat h}_+({\bf p})}p_1\partial_{p_0} {\hat h}_+({\bf p})d\Omega_m
+  \int_{\mathfrak H_m}\overline{{\hat h}_+({\bf p})}p_0\partial_{p_1}{\hat h}_+({\bf p})d\Omega_m \\
& =  \int_{\mathfrak H_m}\overline{{\hat h}_+({\bf p})}p_0\partial_{p_1}{\hat h}_+({\bf p})d\Omega_m
= \frac12\int_{\mathbb R^d}\overline{{\hat h}_+({\bf p})}\partial_{p_1}{\hat h}_+({\bf p})dp \\
& = \frac{i}{2}\int_{\mathbb R^d}x_1 h_+^2({\bf x})d{\bf x} \ ,
\end{align*}
where the first integral in the third line is zero because $\hat h$ does not depend on $p_0$, so
our lemma is proved. 
\end{proof}
\subsection{Time zero momenta}
\begin{lemma}\label{ek2}
Let $k$ be a real function on $\mathbb R^d$ with compact support. We have
\begin{align*}
S(\f_{k_\pi} |\!| \f) 
& = \pi \int_{ x_1 >0}x_1\big(|\nabla k ({\bf x})|^2 +
 m^2 k^2({\bf x})\big)d{\bf x}  \\
 & = \pi \int_{ x_1 >0}x_1 k({\bf x})(\Delta + m^2){ k({\bf x})}d{\bf x}
+ \pi \int_{x_1 = 0} k^2({\bf x}) d{\bf x}
 \ ,
\end{align*}
where the states $\f_{k_\pi} $ and $ \f $ are restricted to $\A(W)$. 
\end{lemma}
\begin{proof}
 As in the proof of Lemma \ref{ek1}, we have
\ben\label{k+}
({k_\pi}_+,{\partial^W_0 {k_\pi}_+}) =  \int_{\mathfrak H_m}p_0p_1\overline{\hat k_+({\bf p})}\hat{ k}_+({\bf p})d\Omega_m
+  \int_{\mathfrak H_m}p^3_0\overline{\hat k_+({\bf p})}\partial_{p_1} \hat{ k}_+({\bf p})d\Omega_m  \  .
\een
The above left integral does not vanishes this time if $k$ in not identically zero on the boundary:
\begin{multline*}
\int_{\mathfrak H_m}p_0p_1\overline{\hat k_+({\bf p})}\hat{ k}_+({\bf p})d\Omega_m
= \frac12\int_{\mathbb R^d}p_1\overline{\hat k_+({\bf p})}\hat{ k}_+({\bf p})d{\bf p}
=-\frac{i}{2}\int_{\mathbb R^d} k_+({\bf x})\partial_{x_1}k_+({\bf x}) d{\bf x} \\
= -\frac{i}{4}\int_{x_1 >0} \partial_{x_1}k^2({\bf x}) d{\bf x}
= \frac{i}{4}\int_{x_1 = 0} k^2({\bf x}) d{\bf x} \ . 
\end{multline*}
The right integral in \eqref{k+} is also computed as in the proof of Lemma \ref{ek1}:
\begin{multline*}
 \int_{\mathfrak H_m}p^3_0\overline{\hat k_+({\bf p})}\partial_{p_1} \hat{ k}_+({\bf p})d\Omega_m  =
 \frac12\int_{ \mathbb R^d}({\bf p}^2+ m^2)\overline{\hat k_+({\bf p} )}\partial_{p_1} \hat{ k}_+({\bf p})d{\bf p} \\
=  \frac{i}{2} \int_{ \mathbb R^d}x_1 k_+({\bf x})(\Delta + m^2){ k_+({\bf x})}d{\bf x} 
=  \frac{i}{2} \int_{ x_1 >0}x_1 k({\bf x})(\Delta + m^2){ k({\bf x})}d{\bf x} \\
= \frac{i}{2}\int_{ x_1 >0}x_1\big(|\nabla k ({\bf x})|^2 +
 m^2 k^2({\bf x})\big)d{\bf x} -
  \frac{i}{2}\int_{ x_1 >0}k({\bf x})\partial_{x_1}k({\bf x})d{\bf x}\\
   =  \frac{i}{2}\int_{ x_1 >0}x_1\big(|\nabla k ({\bf x})|^2 +
 m^2 k^2({\bf x})\big)d{\bf x} -
  \frac{i}4 \int_{ x_1 =0}k^2({\bf x})d{\bf x}
\end{multline*}
that concludes our proof by adding up the two terms. 
\end{proof}
Theorem \ref{first} thus generalises to:
\begin{theorem}\label{second}
Let $h,k\in S(\mathbb R^d)$ be real. The vacuum relative entropy on $\A(W)$ of the coherent state associated with $h_\phi + k_\pi$ is given by 
\[
S(\f_{h_\phi + k_\pi} |\!| \f) = \pi\int_{x_1 > 0} x_1 \Big( h^2(x) +  m^2  k^2(x) + |\nabla k(x)|^2\Big) dx \ .
\]
(Here $\f$ and $\f_{h_\phi + k_\pi}$ are restricted to $\A(W)$.)
\end{theorem}

Note that the second formula in Lemma \ref{ek2} detects the boundary contributions in the relative entropy. 

\subsection{Proof of Lemma \ref{lemmaS+}}
We now show that the relative entropy $S(\f_{f} |_{\A(W)} |\!| \f |_{\A(W)}) $ on $\A(W)$ is given by
\[
S(\f_{f_+} |\!| \f) =
 -2\pi i\int_{\mathfrak H_m}\overline{\hat f_+(p)}\widehat{\partial^W_0 f_+}(p)d\Omega_m \  ;
\]
here $f_+ \equiv f\chi_{{}_W}$ and $\f_{f_+} \equiv \f_f |_{\A(W)}$ as above.

Let $\b_f ={\rm Ad}V(f)$ and $\b_{f_+}$  the automorphism of $\A(W)$ obtained by restricting $\b_f$ to the von Neumann algebra $\A(W)$. 
Then
\[
S(\f_{f_+} |\!| \f) = S(\f \cdot \b^{-1}_{f_+} |\!| \f) \ ,
\]
(states on $\A(W))$. So, as in \cite{L18}, the relative entropy is obtained by differentiating the Connes cocycle (\cite{C73})
\[
S(\f_{f_+} |\!| \f) = i\frac{d}{ds} \f\big( (D \f \cdot \b^{-1}_{f_+} : D\f)_s\big) \big|_{s=0}\ .
\]
Now, we claim that $f_+ - f_{s+}$ belongs to $\H$, namely the Fourier transform of $f_+ - f_{s+}$ belongs to $L^2(\mathbb R^4, d\Omega_m)$, where $f_s = f\cdot\Lambda_W(s)$ is the boost translated of $f$ (although, in general, neither $f_+$ nor $f_{s+}$ belongs to $\H$), with $f = h_\phi$ or $f = k_\pi$, $h,k\in S(\mathbb R^d)$. Indeed, $\hat f_+(p) = O\big(\frac{1}{{\bf p}}\big)$ as $\bf p \to \infty$, so $\hat f_+ - \hat f_{s+} = O\big(\frac{1}{\bf p^2}\big)$ as $\bf p \to \infty$.

Similarly as in \cite{L18}, we have 
\[
(D \f \cdot \b^{-1}_{f_+} : D\f)_{-s/2\pi} 
= V(f_+ - f_{s+})e^{i\Im\int_{\mathfrak H_m}\overline{\hat f_+(p)}\widehat{f_{+}}_s(p )d\Omega_m}\ ; 
\]
then, taking vacuum expectation values and differentiating at $s=0$ as in \cite{L18}, we obtain formula \eqref{S+}:
\begin{multline*}
 i\frac{d}{ds}\f\big((D \f \cdot \b^{-1}_{f_+} : D\f)_{s}\big)\big |_{s=0} 
 =  -2\pi i\frac{d}{ds}
e^{i\Im\int_{\mathfrak H_m}\overline{\hat f_+(p)}\widehat{f_{+}}_s(p )d\Omega_m}\big |_{s=0} \\
= 2\pi \Im \int_{\mathfrak H_m}\overline{\hat f_+(p)}\widehat{\partial^W_0 f_+}(p)d\Omega_m 
= -2\pi i \int_{\mathfrak H_m}\overline{\hat f_+(p)}\widehat{\partial^W_0 f_+}(p)d\Omega_m  \ ,
\end{multline*}
because
\[
\f\big(   V(f_+ - f_{s+})  \big)\big |_{s=0} = 0, \quad 
\frac{d}{ds}\f\big( V(f_+ - f_{s+})  \big)\big |_{s=0} = 0
\]
by formula \eqref{fV}.  

\section{Concluding remarks} Our work may be continued in several directions. The case of charges not localised in the time zero hyperplane, also regarding the behaviour of the relative entropy under null translations, is relevant (see \cite{L18}) and studied in \cite{CLR}. 
Moreover, one could consider the case of a double cone $O$, rather than a wedge $W$, von Neumann algebra. In the massless, finite helicity case, the free field is conformal, the local von Neumann algebras $\A(O)$ and $\A(W)$ are vacuum unitarily equivalent and the modular structure is explicitly known \cite{HL}. One may thus compute relative entropies starting with the general formulas in \cite{L18}. On the other hand, in the infinite helicity case, the von Neumann algebras $\A(O)$ are trivial \cite{LMR18}. The analysis in the massive case would require more knowledge of the modular Hamiltonian. The study of the relative entropy for non-coherent states would be natural, see \cite{LLR} for a first discussion. Finally, it would be interesting to study the case of interacting models.

\bigskip

\noindent
{\bf Acknowledgements.} 

\noindent
We thank the referee for pointing out a trivial calculation error in a previous version of this paper. 

\noindent
We acknowledge the MIUR Excellence Department Project awarded to the Department of Mathematics, University of Rome Tor Vergata, CUP E83C18000100006.

\end{document}